\documentclass{article}
\usepackage{geometry}                
\usepackage{graphicx}
\usepackage{stmaryrd}
\usepackage{amssymb}
\usepackage{amsthm, bm}
\usepackage{amsmath, cancel, centernot }
\usepackage[mathscr]{eucal}
\usepackage{mathtools}
\usepackage{epstopdf}
\title{Discretised Hilbert Space and Superdeterminism}
\author{T.N.Palmer \\ Department of Physics, University of Oxford, UK}
\date{\today}                                          
\makeatletter
\newcommand\be{\@ifstar{\[}{\begin{equation}}}
\newcommand\ee{\@ifstar{\]}{\end{equation}}}
\newcommand\bp{\begin{pmatrix}}
\newcommand\ep{\end{pmatrix}}

\newtheorem{theorem}{Theorem}

\newtheorem*{corollary}{Corollary}
\makeatother
\begin{document}
\bibliographystyle{plain}
\maketitle

\begin{abstract}

In computational physics it is standard to approximate continuum systems with discretised representations. Here we consider a specific discretisation of the continuum complex Hilbert space of quantum mechanics - a discretisation where squared amplitudes and complex phases are rational numbers. The fineness of this discretisation is determined by a finite (prime-number) parameter $p$.  As $p \rightarrow \infty$, unlike standard discretised representations in computational physics, this model does not tend smoothly to the continuum limit. Instead, the state space of quantum mechanics is a singular limit of the discretised model at $p=\infty$. Using number theoretic properties of trigonometric functions, it is shown that for large enough values of $p$, discretised Hilbert space accurately describes ensemble representations of quantum systems within an inherently superdeterministic framework, one where the Statistical Independence assumption in Bell's theorem is formally violated. In this sense, the discretised model can explain the violation of Bell inequalities without appealing to nonlocality or indefinite reality. It is shown that this discretised framework is not fine tuned (and hence not conspiratorial) with respect to its natural state-space $p$-adic metric. As described by Michael Berry, old theories of physics are typically the singular limits of new theories as a parameter of the new theory is set equal to zero or infinity. Using this, we can answer the challenge posed by Scott Aaronson, critic of superderminism: to explain when a great theory in physics (here quantum mechanics) has ever been `grudgingly accommodated' rather than `gloriously explained' by its candidate successor theory (here a superdeterministic theory of quantum physics based on discretised Hilbert space). 
\end{abstract}

\section{Introduction}
\label{intro}

Superdeterminism is an approach to explain the experimental violation of Bell inequalities without invoking nonlocality or indefinite realism. Specifically, a superderministic model negates the Statistical Independence assumption in Bell's theorem \cite{Brans} \cite{Hall:2010} \cite{tHooft:2015b} \cite{HossenfelderPalmer} \cite{DonaldiHossenfelder}: that for a probability $\rho$ on hidden variables $\lambda$ with measurement settings $X$
\be
\label{statind}
\rho(\lambda | X)=\rho(\lambda)
\ee
It is easy to parody a negation of (\ref{statind}). Superficially, it seems to imply grotesque conspiracies that surely no sane physicist would contemplate. For example, superdeterminism appears to require the universe to be initialised in precisely that state which ensures that, billions of years later, the experimenters choose just the right measurement settings (and only those) which lead to a violation of Bell inequalities. Such a picture of the world not only appears to deny experimenters any freedom to choose which experiments they perform -- with some commentators likening it to alien mind control --  it also suggests a theory of the universe that is ludicrously fine tuned. After all, the mere flapping of a butterfly's wings might cause an experimenter to choose a different measurement setting to the one she actually chose. This butterfly, were it allowed to flap, would completely undo the carefully woven interplay between degrees of freedom, supposedly incorporated into the superdeterministic cosmic initial conditions. We must therefore excise all such butterflies in a superdeterministic theory, giving the theory a preposterously conspiratorial flavour.

In any case, so the argument goes, why should one bother to look for alternative explanations for the violation of Bell's inequality to the standard quantum mechanical ones?  Quantum mechanics is, after all, a supremely well-tested theory and has never been found wanting.

Scott Aaronson concludes his excoriating critique of superdeterminism \cite{Aaronson} with a challenge to superdeterminists: to explain when a great theory in physics (in this case quantum mechanics) has ever been `grudgingly accommodated' rather than `gloriously explained' by its candidate successor theory (in this case the supposed superdeterministic theory). If the parody above were accurate, one would indeed be hard pressed to find any sort of answer to Aaronson's challenge. 

However, the parody is not accurate. In this paper, a superdeterministic model of quantum physics is developed whose ensemble characteristics are described by a simple and intuitively appealing concept - the discretisation of Hilbert space \cite{Buniy:2005} \cite{Buniy:2006}. Much of computational physics is based on discretising continuum systems such as fluids, so why not Hilbert space too? In so doing, the parody described above is shown not to be an accurate account of superdeterminism. 

As to why we should seek an alternative explanation for the violation of Bell's inequalities, the answer is simple. Despite over 70 years of intense research, we have not yet been able to synthesise our two great theories of physics: quantum mechanics and general relativity. Perhaps the reason for this is staring us in the face: the notions of indefinite reality and/or the violation of local causality -- the conventional explanations for the violation of Bell inequalities -- are incompatible with the locally causal nonlinear geometric determinism of general relativity.  The fact that it is impossible to signal classical information superluminally in quantum mechanics does not resolve this seeming incompatibility: at the very least there is an unresolved tension here. As to how well tested quantum mechanics really is, it is still an open question as to whether particles whose masses are large enough that their self-gravitation can be detected, behave quantum mechanically. 

Section \ref{dhs} describes the proposed discretisation of Hilbert space. The fineness of the discretisation is governed by a parameter $p$. Here it is shown that a natural form of discretisation (which treats amplitudes and complex phases differently), can take advantage of number theoretic properties of simple trigonometric functions, leading inexorably, no matter how large is $p$, to the concept of counterfactual indefiniteness. In Section \ref{examples} these properties are exploited to discuss two key experiments in quantum physics: the Mach-Zehnder interferometer and the CHSH \cite{CHSH} test of Bell's inequality. In Section \ref{det} a deterministic model of quantum physics is outlined, whose ensemble properties are described by discretised Hilbert space. It is shown how the Uncertainty Principle arises naturally in this superdeterministic model from geometry and number theory. Crucially, it is shown how a natural metric of distance in this model is $p$-adic, and that with respect to this metric the model is neither fine tuned nor conspiratorial. In Section \ref{aar} Aaronson's challenge is discussed. As a potential successor theory, our candidate superdeterministic theory can explain quantum mechanics in exactly the same way that new theories of physics have typically explained old theories since the time of Galileo and Newton -- as a singular limit as a parameter of the new theory is set to zero or infinity. The state space of quantum mechanics arises, not as $p \rightarrow \infty$ but at $p=\infty$. In the Conclusions section we summarise what is wrong with the parody of superdeterminism as presented above. 

It can be noted that this paper extends and simplifies a number of technical results in an earlier paper \cite{Palmer:2020}

\section{Discretising Complex Hilbert Space}
\label{dhs}

We start by returning to the insight that kickstarted the quantum revolution - Max Planck's revolutionary proposal that light energy varies discontinuously in quanta. 

Despite this fundamental postulate, the state space of the theory resulting from Planck's insight  -- quantum mechanics -- is itself a continuous space. Indeed, the continuity of complex Hilbert space can be viewed as one of the most important properties that distinguishes it from classical physics (\cite{Hardy:2004}). We start our model of superdeterminism by applying Planck's proposal once more to the continuum complex Hilbert Space of quantum mechanics: we discretise Hilbert space \cite{Buniy:2005} \cite{Buniy:2006}. At the experimental level this is surely unexceptionable. All experiments which confirm quantum mechanics will necessarily confirm a model of quantum physics based on discretised Hilbert Space, providing the discretisation is fine enough. Conversely, there will be experimental consequences of such a discretisation only if it is possible to probe experimentally beyond the discretisation limit. Below we use the symbol $p$ to denote the fineness of the discretisation -- the larger is $p$ the finer is the discretisation. It would appear $p$ is sufficiently large that it is currently not possible to probe experimentally the \emph{direct} consequences of this discretisation (only indirect ones described in this paper) - but see \cite{HancePalmerRarity}. None of the results in this paper depend on the size of $p$ so long as it is $>12$. In particular, $p$ can be as large as we like, but not infinite. 

A specific form of discretisation is proposed. Specifically, consider a model of quantum physics where complex Hilbert product states of finite dimensional systems of the form
\be
|000\ldots 0\rangle+a_1e^{i \phi_1}|000\ldots 1\rangle+\ldots a_{2^n}e^{i \phi_{2^n}}|11\ldots 1\rangle
\ee
can be deemed physically realistic, only if for all $1 \le j \le 2^n$, $a^2_j \in \mathbb Q$ and $\phi_j / 2\pi \in \mathbb Q$. That is to say, no Hilbert state where squared amplitudes or complex phases are irrational is a realistic state of the proposed model. For a single qubit
\be
\label{qubit}
\cos \frac{\theta}{2} |0\rangle + \sin \frac{\theta}{2} e^{i \phi} |1\rangle
\ee
the discretised Bloch Sphere describes Hilbert vectors of the form (\ref{qubit}), providing $\cos^2 \frac{\theta}{2} \in \mathbb Q \implies \cos \theta \in \mathbb Q$, and $\frac{\phi}{2\pi} \in \mathbb Q$. More specifically, we invoke the finite discretistion
\be
\label{dis}
\cos^2  \frac{\theta}{2} = \frac {m}{p}, \ \ \ 
\frac{\phi}{2\pi} = \frac{n}{p} 
\ee
For reasons discussed below $p>12$ is a prime number. For experimental reasons (discussed above), we assume $p \gg 0$. Here $n$ and $m$ are positive integers with $0 \le n \le p $, $0 \le m \le p$. Note that the discretisation can be made as dense as we like on the Bloch sphere by making $p$ sufficiently large (and as Euclid first showed, there are infinitely many prime numbers).

There are two key reasons why such a discretisation is natural. Firstly, real unit Hilbert vectors with rational square amplitudes can describe finite systems (including classical systems) under uncertainty. Consider a shuffled pack of face-down cards. The state of the top card can, for example, be described by the unit Hilbert state
\be
\label{hilbert}
\sqrt{\frac{1}{52}} | \text{ace of clubs}\rangle + \sqrt{\frac{51}{52}} |\text{not ace of clubs}\rangle
\ee  
This state can be represented as a unit vector in 2D Euclidean space because of Pythagoras's Theorem \cite{Isham}. Interpreting these squared amplitudes as probabilities (Born's Rule) is consistent with the fact that the probability of the top card being either the ace of clubs or not the ace of clubs is equal to one. 

Conversely, a natural way to discretise the unit complex numbers is through the set $\{e^{2\pi i n/p}\}$ of $p$th roots of unity, a cyclic multiplicative group of order $p$. In Section \ref{det} we represent the elements of such a group by the set of cyclic permutations of a bit string with $p$ elements, where $p >12$ is a prime number. A crucial number theorem for all that follows is what we refer to as Niven's Theorem
\begin{theorem}
\label{niven}
 Let $\phi/2\pi \in \mathbb{Q}$. Then $\cos \phi \notin \mathbb{Q}$ except when $\cos \phi =0, \pm \frac{1}{2}, \pm 1$. \cite{Niven, Jahnel:2005}
\end{theorem}
Based on Niven's Theorem we derive the `Impossible Triangle Corollary' which is central to the results below
\begin{figure}
\centering
\includegraphics[scale=0.3]{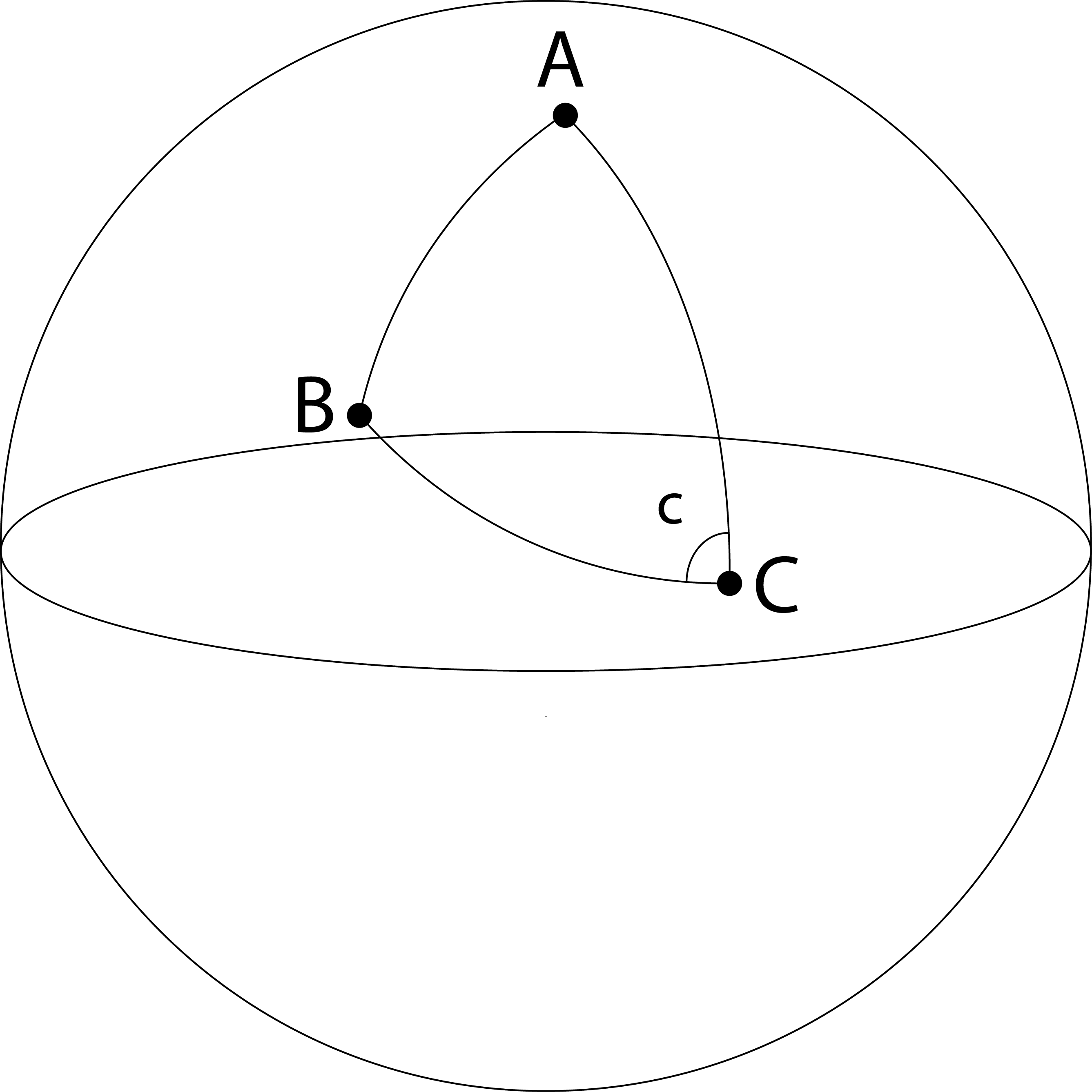}
\caption{\emph{Let $\triangle ABC$ denote a non-degenerate triangle where the cosines of the angular lengths of $AC$ and $BC$ are rational numbers, and the angle $c$ subtended at $C$ is rational. By Niven's theorem the cosine $\cos AB$ of the angular length of $AB$ is typically irrational, with a small number of exceptions. All of these exceptions can be ruled out if $\cos AB$ is of the form $2 \frac{m}{p} -1$, c.f. (\ref{dis}), where $p>12$ is a prime number and $0 \le m \le p$ is an integer.}}
\label{ABCtriangle}
\end{figure}

\begin{corollary}
\label{triangle}
Let $\triangle ABC$ denote a non-degenerate triangle on the unit sphere, with vertices $A$, $B$ and $C$ (see Fig \ref{ABCtriangle}). We suppose the angles subtended at $A$, $B$ and $C$ are of the form $2\pi n/p$ for $1<n<p$. Then it is impossible for the angular lengths ($AB$, $BC$, $CA$) of all three sides of the triangle to have rational cosines.
 \end{corollary}
\begin{proof}
We make use of the cosine rule for spherical triangles:
\be
\label{cosinerule}
\cos AB = \cos AC \cos BC + \sin AC \sin BC \cos c
\ee
Let us suppose that all of $\cos AB$, $\cos BC$ and $\cos AC$ are rational. Then $\sin AC \sin BC \cos c$ must be rational. Hence, squaring, $(1- \cos^2 AC)(1-\cos^2 BC) (2 \cos 2c - 1)$ must be rational. Hence $\cos 2c$ must be rational. Since, by assumption, $c$ is rational, then by Niven's Theorem $\cos 2c$ must equal either $0$, $\pm \frac{1}{2}$ or $\pm 1$. All of these exceptions can be ruled out, either because of the form $c = 2\pi n/p$ for prime $p>12$, or because $\triangle ABC$ is non-degenerate (so that $c \ne 0$ or $\pi$ precisely). 
\end{proof}

\section{Consequences of the Discretisation}
\label{examples}
Consider two immediate consequences of Niven's Theorem. 
\subsection{Interpreting the Mach-Zehnder Interferometer}
Consider a standard Mach-Zehnder interferometer (see Fig \ref{MZ}) where $\Delta \phi$ denotes a phase difference between the two arms of the interferometer. We will assume $\Delta \phi$ is not equal to zero \emph{precisely}: ubiquitous astrophysical gravitational waves will ensure this, as LIGO has demonstrated. From standard quantum mechanics, the Hilbert vector associated with the output of the interferometer is 
\be
\psi_I= \cos \frac{\Delta \phi}{2} |A\rangle +\sin \frac{\Delta \phi}{2}|B\rangle
\ee
By contrast, if we remove the second half-silvered mirror so now the device simply measures which arm a particular particle took after passing through the first beam splitter (Fig \ref{MZ}b) then the relevant Hilbert vector is
\be
\psi_{W}= |B\rangle + e^{i \Delta \phi} |A\rangle
\ee
Let us try to interpret this in our discretised Hilbert space. Suppose Alice performs an experiment with a Mach-Zehnder device. She has to decide whether to perform the experiment with the second half-silvered mirror in place, or removed. If she performs an interferometric experiment, then we can conclude that $\cos \Delta \phi $ must be rational. On the other hand, if she performs a which-way experiment, then $\Delta \phi$ must be rational. We can now invoke Niven's theorem. If $\Delta \phi \ne 0$ precisely, and if $\Delta \phi$ is of the form $2\pi n/p$ then we can disregard all exceptions to Niven's theorem and deduce that it is not simultaneously possible for $\cos \Delta \phi$ and $\Delta \phi$ to be rational. That is to say,  it is not simultaneously possible to perform an interferometric and a which-way experiment on the same particle. This of course, is exactly what quantum mechanics says, except that here we invoke number theoretic arguments to explain this property of the quantum system. 
\begin{figure}
\centering
\includegraphics[scale=0.3]{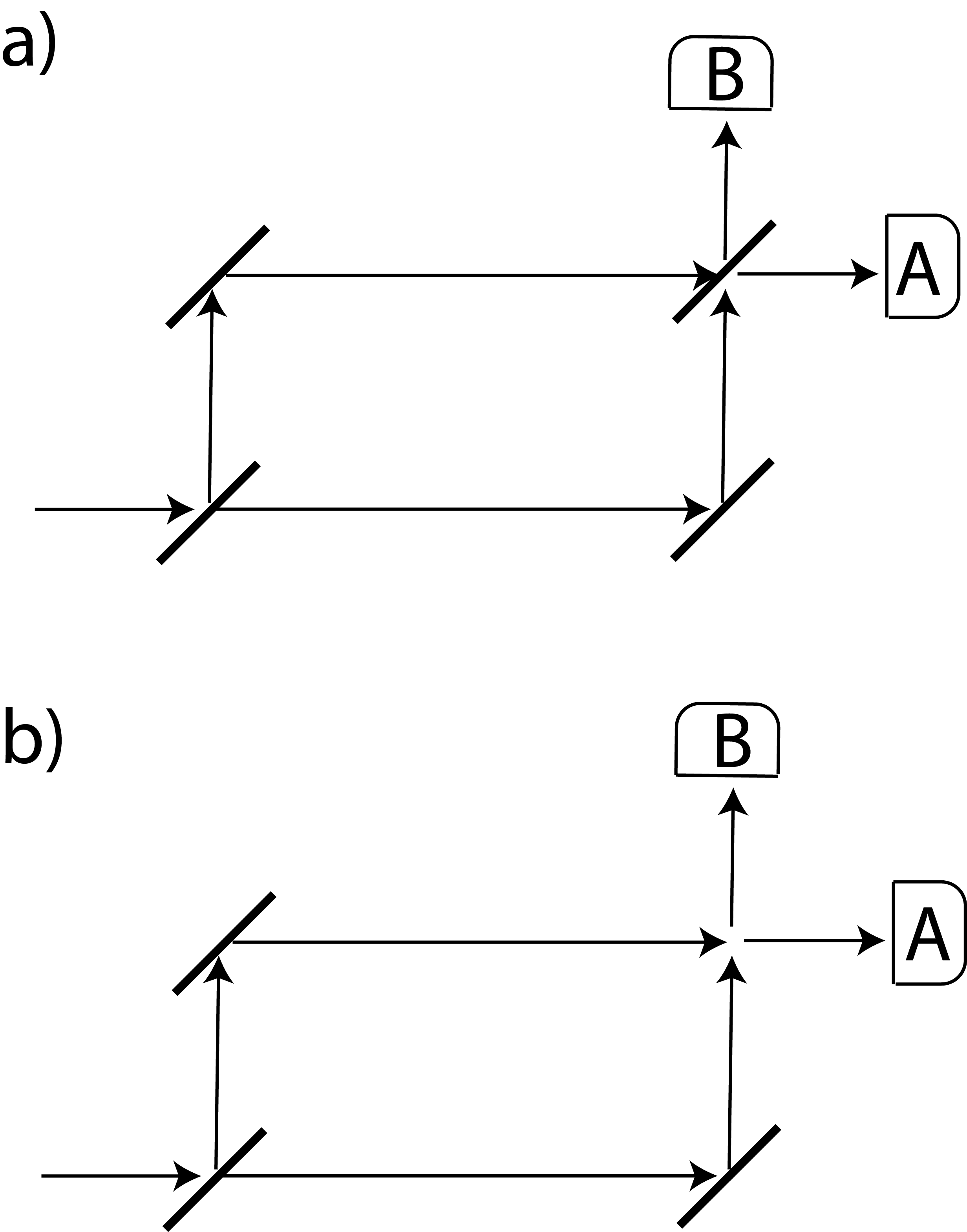}
\caption{\emph{a) A Mach-Zehnder interferometer configured to make an interferometric measurement. b) As in a) but with the second half-silvered mirror removed, thus configured to make a which-way experiment.}}
\label{MZ}
\end{figure}

By saying that two conceivable measurements are not simultaneously possible, we are effectively claiming that a counterfactual world where Alice performed a which-way experiment on a particle, where in the real world she performed an interferometric measurement on that same particle, is inconsistent with our discretisation of Hilbert space. Essentially the counterfactual state corresponds to a Hilbert vector that lies off the skeleton of discretised Hilbert states. It doesn't matter how fine the discretisation is, the counterfactual will never lie on the discretised skeleton. 

However, the problem of fine tuning raises its head. From the Euclidean geometric perspective,  a point on the circle associated with a rational angle can lie arbitrarily close to a point associated with an irrational angle. If we are to make a plausible model of quantum physics based on this discretisation of Hilbert space, the model cannot be based on Euclidean geometric concepts. We return to this point in Section \ref{det} and the Conclusions section.  

\subsection{Interpreting Bell's Theorem}
\label{Bell}

Making use of the Impossible Triangle Corollary, the proposed discretisation of Hilbert Space can form the basis for an explanation of Bell's inequality without having to invoke nonlocality (i.e., without having to violate local causality) or indefinite reality. 

Consider a CHSH experiment where $X=0,1$ denote Alice's choices of spin measurement, and $Y=0,1$ denote Bob's choices. Bell's inequality is, at heart, a number theoretic inequality that is an inevitable property of a look-up table for the spins of either Alice or Bob's particles relative to the four measurement orientations $X=0$, $X=1$, $Y=0$ and $Y=1$. Such a look-up table would be implied by some computable function $Sp(\lambda, X)$ from a classical hidden-variable theory. See Fig \ref{lookup}a. 
\begin{figure}
\centering
\includegraphics[scale=0.6 ]{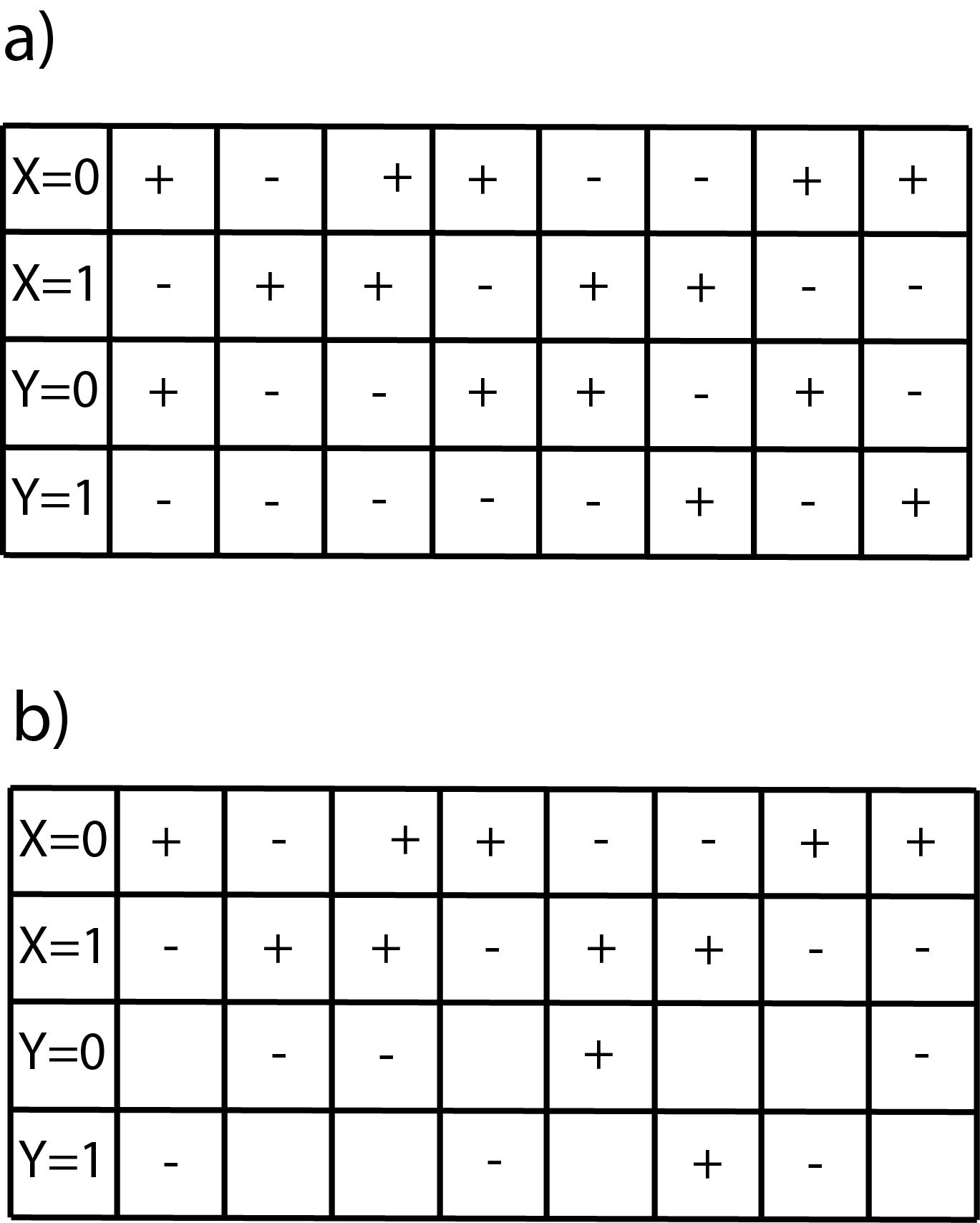}
\caption{\emph{a) An example of a putative look-up table associated with a conventional hidden-variable theory, applied to a CHSH experiment. Each column represents one of Alice's ensemble of 8 particles, each with its own hidden variables. The rows represents the four measurement orientations $X=0$, $X=1$, $Y=0$ and $Y=1$. Here the value '+' denotes spin up, '-' denotes spin down. b) The same look-up table in a model based on discretised Hilbert space. Because of the Impossible Triangle Corollary, there will always be one blank entry for each column, corresponding to a counterfactual measurement for which there is no + or - outcome. Hence it is impossible to establish a look-up table and Bell's inequality cannot be established. Because this model is based on an arbitrarily fine discretisation of complex Hilbert space, this model violates Bell inequalities as does quantum mechanics for large enough $p$. This discussion makes no reference to causality, local or otherwise.}}
\label{lookup}
\end{figure}
The key result proved below is that such a look-up table cannot be inferred from our discretised Hilbert space model. 

For each column, we can immediately fill in one row of the putative look-up table using the observed outcome of the CHSH measurement that was actually performed. For example, suppose (without loss of generality) Alice measures the fourth of her particles in the $X=0$ direction as spin up. Then that column will have a + value in the $X=0$ row, as shown in Fig \ref{lookup}.

The fact that Bob performs a measurement on his entangled particle allows Alice to infer a second + or - entry in her look-up table. For example, if Bob chooses $Y=0$ and obtained +, then (by the singlet nature of the entangled state), Alice can infer that had she measured her particle in the $Y=0$ direction, she would have measured -.  This is an inference made from a measurement that Alice did not make, but might have made: it is an inference from a counterfactual measurement. 

The orientation of the measuring devices associated with any one of the labels $X=0$, $X=1$, $Y=0$ and $Y=1$ will not be \emph{precisely} the same from one column (one entangled particle pair) to the next. There will inevitably be small fluctuations that cannot be controlled by the experimenter between one measurement (one column) and the next, e.g. from gravitational waves or hand trembling to use a phrase by John Bell. However, when Alice and Bob perform a particular pair of measurements on a particle pair, we can assume their choice of $X=0$, $X=1$, $Y=0$ and $Y=1$ implies some particular precise measurement orientations, denoted by points on the celestial sphere (see Fig \ref{CHSHtriangle}). The precise position of these points will vary slightly from one particle pair (look-up table column) to the next. However, when we come to evaluate counterfactual measurements, such as the one above, i.e. ask what the outcome would have been had Alice measured in Bob's chosen direction, we are referring to Bob's \emph{precise} direction, and not just some nominal direction. 

Let us suppose that the angular distance between Alice and Bob's measurement orientation on the celestial sphere is equal to $\theta$. In our discretised model of Hilbert space (where the entangled singlet state is evaluated in Alice and Bob's actual measurement basis) we require $\cos \theta$ to be rational. Since rational cosines are dense on the sphere (and the discretisation is as fine as we like for large-enough $p$), it is in practice no constraint to assume that this condition can be met. 

Of course, the converse is also true. Bob can similarly assert counterfactually that had he measured in Alice's direction, he would have obtained the opposite to what Alice obtained. 

\begin{figure}
\centering
\includegraphics[scale=0.3]{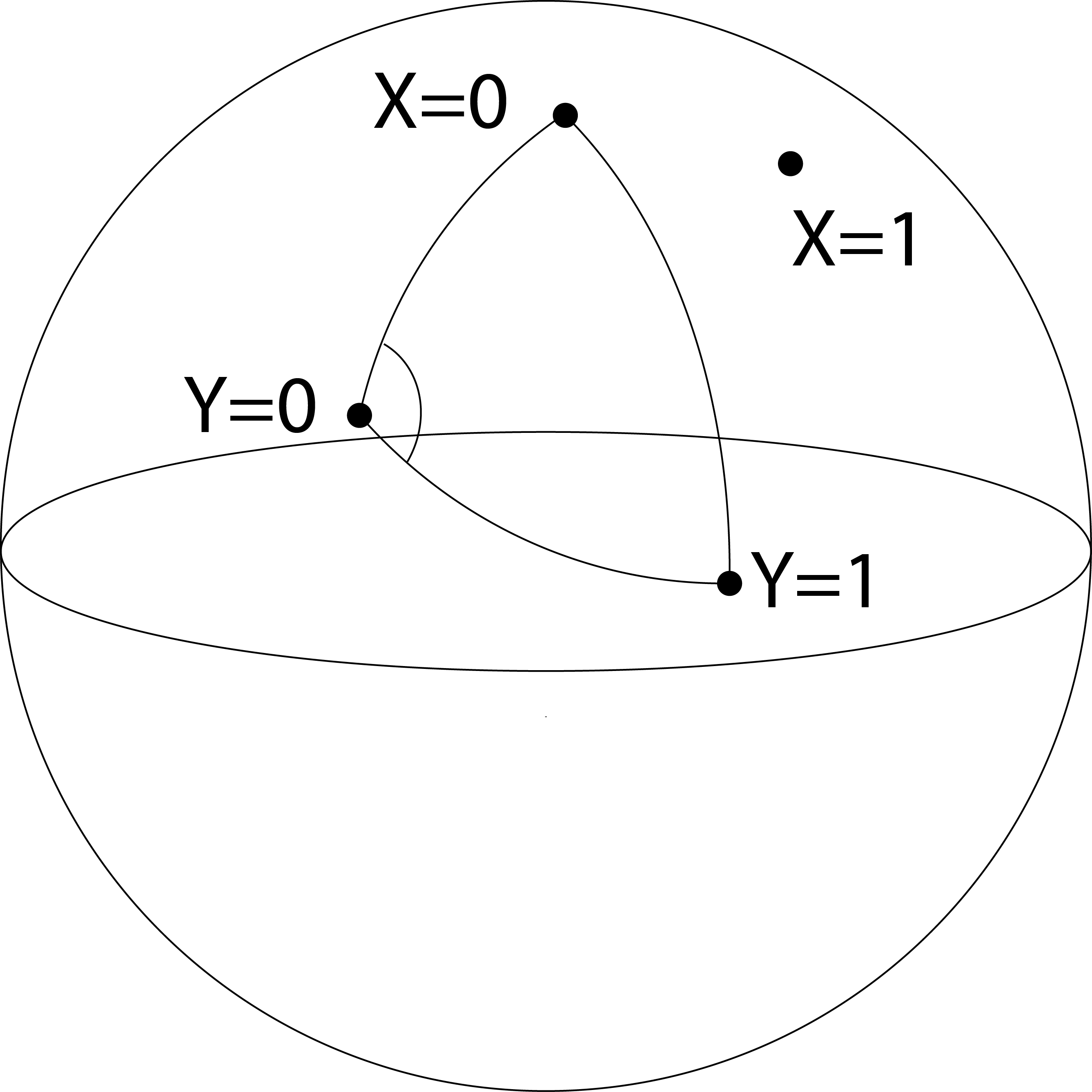}
\caption{\emph{We represent the four choices of measurement orientation in a CHSH experiment, for a particular entangled particle pair, by four points on the unit celestial sphere. We can imagine, without loss of generality, that Alice chooses $X=0$ and Bob $Y=0$. Alice can infer that had she measured in Bob's direction she would have measured the opposite of what Bob measured. However, as discussed in the text, by the Impossible Triangle Corollary she cannot infer that had she measured in Bob's counterfactual direction, she would have measured the opposite of what Bob would have measured. As a result, it is impossible to infer that a hidden-variable theory whose statistics are described by the discretisation of Hilbert Space described above, must satisfy Bell's inequality. Indeed precisely because the discretisation can be arbitrarily fine, such a model must violate Bell's inequality exactly as does quantum theory.}}
\label{CHSHtriangle}
\end{figure}

On this basis we can fill with +s and -s two rows in the fourth column (and by implication any column) in Alice's look-up table. What about the other two rows? In the example where Alice measures with $X=0$, the $X=1$ direction is counterfactual. As before, it is no constraint in practice to assume that the cosine of the angular distance between $X=0$ and $X=1$ is rational (noting again that the precise value of this rational may well vary from one column to the next). Hence we can in principle fill in a third column of Alice's look-up table with either a + or a - (even though we may not know what the correct value is). 

Let's now turn to $Y=1$. In the example where Bob actually measures in the direction $Y=0$, the direction $Y=1$ is counterfactual for Bob. Again Bob can assert that had he measured in the $Y=1$ direction he would have got some definite +/- output, since it is again no constraint to assume that the cosine of the angular distance between $Y=0$ and $Y=1$ is rational. 

We need to note that in order that \emph{both} Alice's actual direction $X=0$ and Bob's counterfactual direction $Y=1$ are both permissible counterfactuals relative to Bob's actual direction $Y=0$, then the angle subtended at the vertex $Y=0$ must itself be a rational angle. Again, by the density of rational angles, there is no constraint preventing that. 

All good so far. But we are set up for the killer blow. Let us return to Alice's look-up table. We have filled in three entries in the fourth column. However, in order to infer a value in the look-up table for the $Y=1$ row, Alice must be able to assert the `double' counterfactual C: `if I had measured in the $Y=1$ direction (Bob's counterfactual direction), I would have got the opposite of what Bob would have got had he measured in this direction.' But now we can invoke the Impossible Triangle Corollary - see Fig \ref{CHSHtriangle}. We have already specified that the angular distances of the sides $(X=0, Y=0)$ and $(Y=0, Y=1)$ of our triangle have rational cosines, and that the angle subtended at the common vertex $Y=0$ is rational. \emph{It is therefore impossible for the third side $(X=0, Y=1)$ to have a rational cosine}. Hence Alice's counterfactual assertion C is neither true nor false, it is simply undefined. It is inconsistent with the assumed discretisation of Hilbert space. This means we cannot fill in the row $Y=1$ for the fourth column of the look=up table. Indeed, for each column of the look-up table we are unable to fill in all four rows. As such our model is not constrained by Bell's inequality. Indeed, since the model is based on an arbitrarily fine discretisation of complex Hilbert space, we can claim that it violates Bell inequalities exactly as does quantum mechanics (respecting, for example, the Tsirelson bound). 

Note that we have not violated local causality here. This is easy to see simply because \emph {we have never once mentioned the notion of causality}. This may perhaps seem surprising. However, it is worth quoting Bell himself here \cite{Bellb}:
\begin{quote}
I would insist here on the distinction between analysing various physical theories, on the one hand, and philosophising about the unique real world on the other hand. In this matter of causality, it is a great inconvenience that the real world is given to us once only. We cannot know what would have happened if something had been different. We cannot repeat an experiment changing just one variable the hand of the clock will have moved, and the moons of Jupiter. Physical theories are more amenable in this respect. We can calculate the consequences of changing free elements in the theory, be they only initial conditions, and so can explore the causal structure of the theory. I insist that [Bell's Theorem] is primarily an analysis of certain kinds of physical theory. 
\end{quote}
Here Bell is saying that, fundamentally, Bell's Theorem is primarily an analysis of the degrees of freedom in putative theories of quantum physics. 

Neither have we invoked some kind of indefinite reality. Rather, as shown below, we have violated the Statistical Independence assumption in Bell's Theorem - but in a natural way, and not in terms of the grotesque parody described in the Introduction. 

\section{Towards a Superdeterministic Theory of Quantum Physics}
\label{det}

Here, motivated by the fractal invariant sets of chaotic dynamical systems, we outline a geometric construction for a deterministic model that can generate the discretised Hilbert space discussed above. Consider the trajectory of a dynamical system as a fractal curve, akin to a length of rope (strands around strands around strands...) as shown in Fig \ref{invariant}a. Like all fractals we can define this fractal trajectory in an iterative way. To first order the trajectory resembles a solid tube of radius proportional to $\hbar$. In the limit $\hbar=0$, this tube collapses to a one-dimensional curve (of elementary classical physics). This tube therefore can describe the semi-classical physics of Maxwell and Dirac fields. To second order, we find that the tube comprises a helix of $p$ smaller tubes, each with radius proportional to $\hbar^2$.  Properties of the ensemble of such helical tubes is described by the discretised Hilbert space above. Any single element of this helix is again describable by semi-classical physics. However, it too comprises by a helix of $p$ trajectories, and so on. This picture of the laws of physics continually iterating between classical and quantum descriptions can readily describe such mysteries as wave-particle duality, as will be described in more detail elsewhere. 
\begin{figure}
\centering
\includegraphics[scale=0.3]{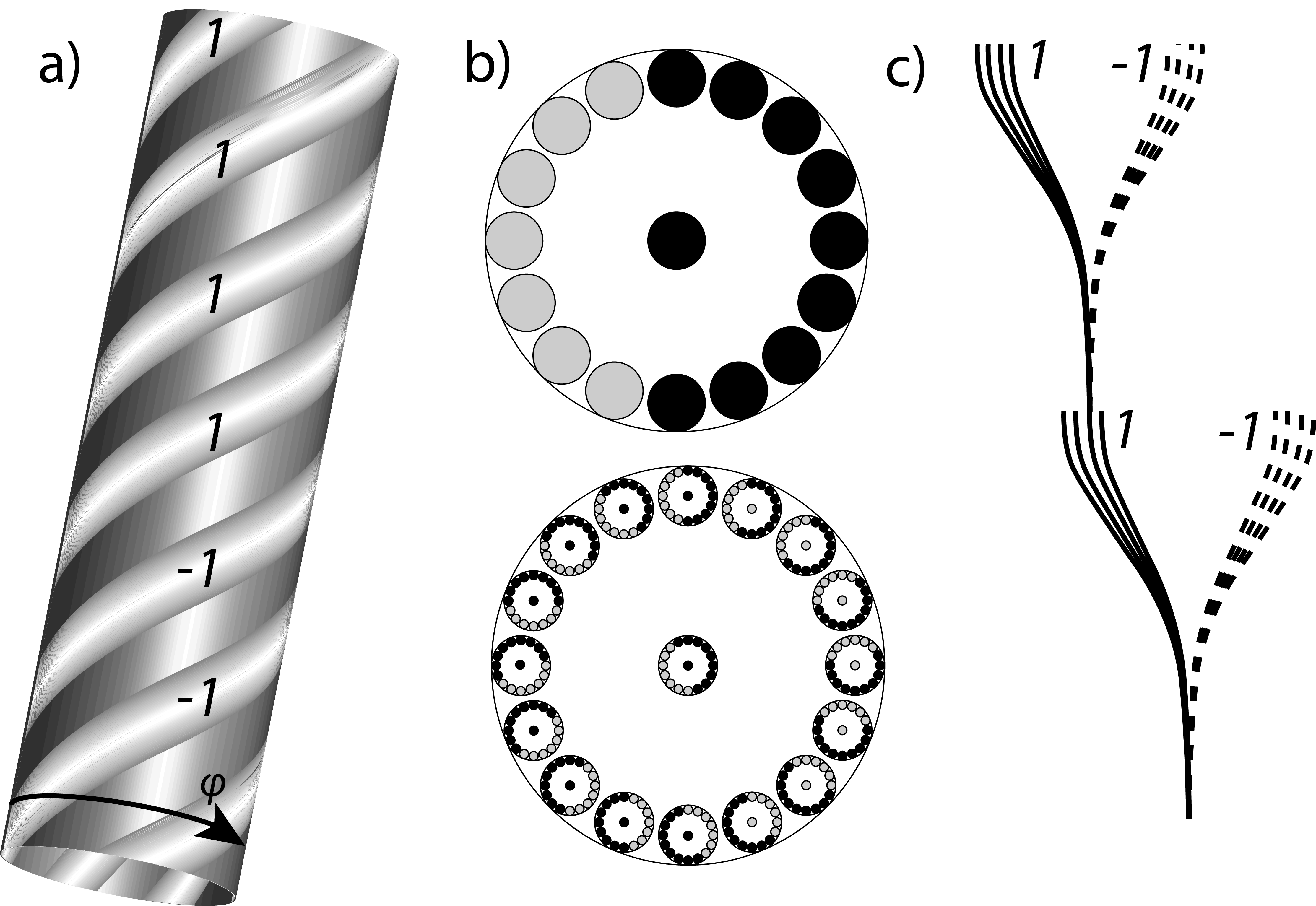}
\caption{\emph{The state-space geometric structure of the fractal invariant set of a deterministic dynamical system which generates the type of discretised Hilbert space described above. a) each trajectory segment at any level of fractal iterate actually comprises a helix of trajectory segments. b) A cross-section through a helix (shown at two iterate levels) reveals a fractal structure homeomorphic to the set of $p$-adic integers. c) the process of quantum measurement is associated with a divergence and nonlinear clustering into measurement eigenstates. Divergence and clustering is shown for two levels of fractal iterate. The trajectories iterates are labelled `1' or `-1 or coloured black/grey, according to the cluster to which a trajectory evolves. Here is illustrated the situation for $p=17$. In practice $p$ is presumed to be much larger.}}
\label{invariant}
\end{figure}
If we take a cross section of this fractal trajectory, we will obtain something topologically like that shown in Fig \ref{invariant}b) for $p=17$ (for illustration). This cross section is equivalent to a Cantor set with $p$ iterated pieces. Importantly, the set of $p$-adic integers is homeomorphic to such a Cantor set \cite{Katok}. This gives us a strong clue about the type of deterministic dynamical system needed to explain our discretised Hilbert space: it will be a $p$-adic dynamical system based on mappings of $p$-adic integers \cite{Dragovich} \cite{Khrennikov}. This has important implications for the notion of conspiracy and fine tuning, discussed below.  

Let us fix on the second-order iteration. Imagine that as a quantum system interacts with the measuring device, the $p$ trajectories diverge exponentially from one another in an initially linear phase, but ultimately cluster together nonlinearly into a number $N < p$ of discrete clusters. These clusters correspond to measurement eigenstates. For simplicity, Fig \ref{invariant}c) shows a situation with just $N=2$ clusters which we label as $1$ and $-1$,, at two separate levels of fractal iteration. 

We can label each of the $p$ trajectories before divergence with the cluster label to which the trajectory evolves in the future, as shown in Fig \ref{invariant}a). This does not imply or invoke retrocausality. It is merely a method to label a trajectory based on its future evolution (much as we would label points in space-time as lying inside or outside a black-hole event horizon -- the latter being defined as the boundary of null geodesics that escape to future null infinity). In the case where there are two measurement eigenstates (a qubit), the $p$ trajectories can be represented by a bit string. Indeed, let $S(m)$ denote the bit string
\be
\label{bitstring}
\{\underbrace{1, 1,\ldots 1}_{m}, \underbrace{-1, -1, \ldots -1}_{p-m}\}
\ee
and let $\zeta$ denote a cyclic permutation of $S(m)$ so that $\zeta^p S(m)=S(m)$. The group $\{\zeta^n\}$ is a representation of the multiplicative group $\{e^{2 \pi i n/p}\}$ of $p$th roots of unity. We can therefore write 
\be
S(n,m) \equiv \zeta^n S(m) \equiv e^{2\pi n/p} \; S(m) 
\ee 
Indeed using (\ref{dis}), we can define $S(\phi, \theta)=S(n,m)$ so that each point on the discretised Bloch sphere, and hence any discretised Hilbert state as described by (\ref{qubit}), is associated with a bit string (i.e. an ensemble of symbolically labelled trajectories) of the form $S(n,m)$ modulo a global permutation (corresponding to a global phase factor). Now let $<\ldots >$ denotes the mean over the bit string and $\Delta S$ denotes the standard deviation of the bit string. Then, from (\ref{bitstring})
\be
<S(\phi, \theta)> = \frac{1}{p} [m -(p-m)] = 2 \frac{m}{p}-1 = \cos \theta
\ee
using (\ref{dis}). Similarly
\be
(\Delta S)^2=<S^2>-<S>^2=1-\cos^2 \theta
\ee
so that 
\be
\Delta S = |\sin \theta|
\ee
As an application of this symbolic bit-string representation, we derive the quantum Uncertainty Principle for spin-1/2 particles using spherical triangles and the Impossible Triangle Corollary. In particular, consider a generic point $p$ on the Bloch Sphere, whose co-latitudes with respect to three polar directions $p_z$, $p_x$ and $p_y$ are $\theta$, $\theta'$ and $\theta''$. See Fig \ref{uncertainty}
\begin{figure}
\centering
\includegraphics[scale=0.3]{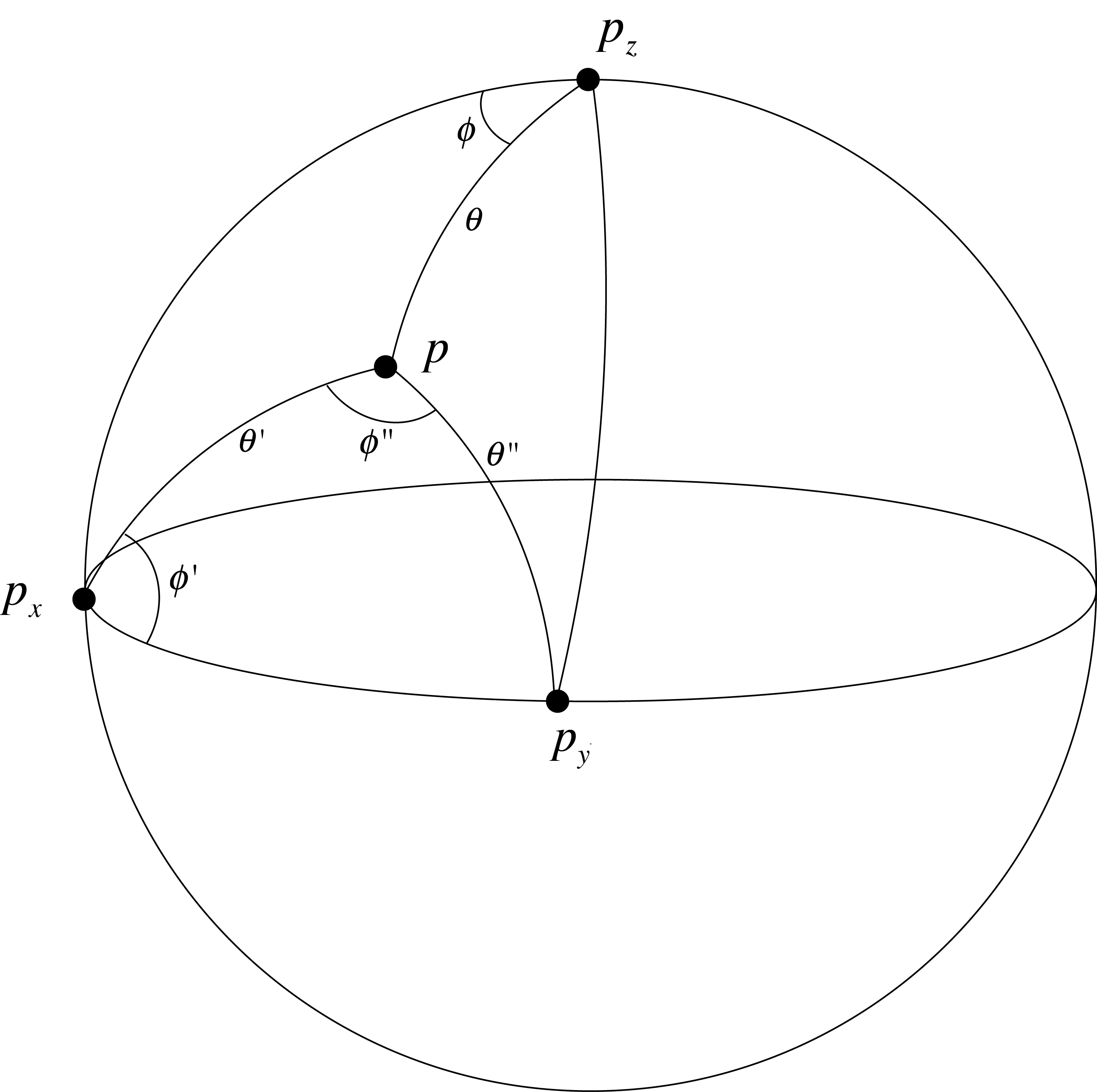}
\caption{\emph{The Uncertainty Principle from geometric and number-theoretic considerations. Using the sine and cosine rule for spherical triangles we show that $\Delta S_x \Delta S_y \ge | <S_z>|$, where the terms are defined in the text. Using the Impossible Triangle Corollary, we show that it is impossible to determine spin simultaneously relative to any two different directions.}}
\label{uncertainty}
\end{figure}

First apply the cosine rule on the spherical triangle $\triangle p p_x p_z$ so that
\be
\cos \theta = \cos \theta' \cos \frac{\pi}{2} + \sin \theta' \sin \frac{\pi}{2} \sin \phi'
\ee
which implies
\be
\label{up1}
\cos \theta = \sin \theta' \sin \phi'
\ee
Now apply the sine rule on the spherical $\triangle p p_x p_y$ so that
\be
\frac{\sin \phi'}{\sin \theta''}=\frac{\sin \phi''}{\sin  \frac{\pi}{2}}= \sin \phi''
\ee
which implies
\be
\label{up2}
 |\sin \phi'| \le |\sin \theta''| 
\ee
Putting (\ref{up1}) and (\ref{up2}) together we have
\be
|\sin \theta'| |\sin \theta''| \ge |\cos \theta|
\ee
or
\be
\label{up3}
\Delta S_x \Delta S_y \ge | <S_z>|
\ee
If instead of (\ref{bs}) the bit strings are dimensional of the form
\be
\label{bs}
S(m)=\{\underbrace{\frac{\hbar}{2}, \frac{\hbar}{2},\ldots \frac{\hbar}{2}}_{m}, \underbrace{-\frac{\hbar}{2}, -\frac{\hbar}{2}, \ldots -\frac{\hbar}{2}}_{p-m}\}
\ee
then (\ref{up3}) can be written
\be
\Delta S_x \Delta S_y \ge \frac{\hbar}{2} | <S_z>|
\ee
which is the standard form for the Uncertainty Principle in quantum mechanics. 

But of course there is much more to the Uncertainty Principle. It also encompasses the notion of quantum complementarity, that if the state on the Bloch Sphere at the point $p$ cannot simultaneously describe a measurement outcome with respect to the two directions $p_z$ and $p_x$ say. In the current superdeterministic theory, this idea is expressed by showing that if $\cos \theta$ is rational, then $\cos \theta'$ cannot be rational. 

This can be shown from the Impossible Triangle Corollary applied to the spherical triangles $\triangle p p_z p_x$. Here we do not need to assume that the angular distance between $p_z$ and $p_x$ is precisely a right angle. It's merely enough to assume that the cosine of the angular distance between $p_z$ and $p_x$ is rational. By the Corollary, if $\cos \theta$ is also rational, then $\cos \theta'$ cannot be rational. Hence the point $p$ cannot be simultaneously be associated with a measurement in the $z$ direction and the $x$ direction. A similar argument applies to the triangle $\triangle p p_z p_y$. 

Based on geometric notions of fractal invariant sets, we have described geometrically a deterministic model whose ensemble properties generate the discretisation of Hilbert Space described in Section \ref{dhs}. Ultimately, this invariant set must describe the evolution of the whole universe. That is to say, motivated by nonlinear dynamical systems theory, we assume the universe can be considered a deterministic dynamical system evolving precisely on its fractal invariant set. This means that the most primal quantity in physics is not some assumed cosmological initial condition, to be mapped forward by a dynamical evolution law, but as a `timeless' geometry in state space. 

As discussed, the bit strings discussed above describe labels for a fractal geometry homeomorphic to the set $\mathbb Z_p$ of $p$-adic integers. This has important implications because the $p$-adic metric has different properties to the more familiar Euclidean metric. For example, the $p$-adic distance between two points in $\mathbb Z_p$ is always less than unity. By contrast the distance between two points in the set of $p$-adic numbers $\mathbb R_p \supset \mathbb Z_p$ can exceed $p$. Based on this we can define a metric on the Euclidean space in which our fractal trajectories are embedded: for any two points which lie on the fractal trajectory, the distance between such points is $p$-adic and necessarily less than unity (and decreases by a power of $p$ for each increasing fractal iterate). In all other situations the distance between two points is equal to $p$. It is easy to show this satisfies the conditions for a metric. If $p$ is large, then two points which do not both lie on the invariant set are necessarily distant from one another, even though they may appear close from a Euclidean perspective. This has major implications for the issue of conspiracy and fine tuning, discussed in the Conclusions section. 

The theory described here is referred as invariant set theory \cite{Palmer:2020}. Invariant set theory can be described using the language of hidden variables. That is to say, we can describe some measurement outcome by the form $Sp=Sp(\lambda, X, Y)$ where $\lambda$ describes a particle's hidden variables and $X$ denotes measurement settings. The theory is local because based on the discussion in Section \ref{Bell} we can write $Sp=Sp(\lambda, X)$ for Alice's particles. If we relate this to the form $S=S(\phi, \theta)$, then we can associate $\lambda$ with one of the $p$ trajectory segments in Fig \ref{invariant}. In space-time, $\lambda$ describes a particle relative to all the other particles in the universe. 

Using hidden variable language we can show that invariant set theory violates the Statistical Independence assumption and hence is superdeterministic - though we now prefer to use the word `supermeasured' \cite{Hance:2022} to describe such a theory - to distinguish this type of theory from those that are genuinely conspiratorial. 

For example, let $X=1$ denote the Mach-Zehnder set up when Alice chooses to perform an interferometric experiment, and $X=0$ the Mach-Zehnder set up when she chooses to perform a which-way experiment. Then Niven's theorem implies
\be
\label{SuperD}
\rho(\lambda | X) \ne 0 \implies \rho(\lambda | X')=0
\ee
where $X'=1-X$. A theory where (\ref{SuperD}) holds violates Statistical Independence because it implies that $\rho(\lambda | X) \ne \rho(\lambda)$. In the case of the CHSH experiment, using the Impossible Triangle Corollary we have shown that
\be
\label{SD}
\rho(\lambda | X Y) \ne 0 \implies \rho(\lambda | X Y')=0
\ee
where $Y'=1-Y$. Again this implies $\rho(\lambda | X Y) \ne \rho(\lambda)$ which is a violation of the Statistical Independence assumption in Bell's Theorem. 

\section{The Aaronson Challenge}
\label{aar}

We can finally discuss the Aaronson Challenge \cite{Aaronson}. Let us first state the challenge in the form Aaronson posed it at the end of his blogpost:

\begin{quote}
Here's my challenge to the superdeterminists: when, in 400 years from Galileo to the present, has such a gambit ever worked? Maxwell's equations were a clue to special relativity. The Hamiltonian and Lagrangian formulations of classical mechanics were clues to quantum mechanics. When has a great theory in physics ever been grudgingly accommodatedby its successor theory in a horrifyingly ad-hoc way, rather than gloriously explained and derived?
\end{quote} 

This raises the question: how typically do new theories accommodate the old theories that they replace. The renowned theoretical physicist Michael Berry has addressed this question and concluded \cite{Berry} that typically an old theory is a singular limit of a new theory as some parameter of the new theory is set equal to infinity or zero. The behaviour of the old theory does not gradually emerge from the new theory as this parameter tends to zero or infinity. Berry cites as examples how the old theory of ray optics emerges from Maxwell theory, or as the old theory of thermodynamics emerges from statistical mechanics. Newtonian gravity can similarly be thought of as a singular limit of general relativity - Newtonian gravity is not a theory of space-time geometry for slowly moving small masses. And indeed classical physics is a singular limit of quantum mechanics as $\hbar$ is set equal to zero. 

As a graphic example of a singular limit, Berry imagines biting into an apple and finding half a maggot -- unpleasant as it implies you have eaten the other half. Worse still is biting into an apple and finding a quarter of a maggot as it implies you have eaten three quarters. Continuing this, then in a limiting sense biting into an apple and finding no maggot would be the worst experience of all. But obviously it's not, since by and large we don't tend to find maggots in our apples. 

Hence, if this process is typical, and if the successor theory to quantum mechanics is indeed superdeterministic, we would expect quantum mechanics to arise as a singular limit of the successor theory. 

This is exactly the situation with respect to the parameter $p$. In our superdeterministic theory $p$ is finite but can be as large as we like. The fractal gaps exist no matter how large is $p$, and these fractal gaps can be invoked to account for the inconsistency of certain key counterfactuals needed, for example, to explain deterministically the violation of Bell's inequality. These fractal gaps only disappear \emph{at} the limit $p = \infty$. At this limit, and only at this limit, the state space of the discretised set of Hilbert states becomes an (arithmetically closed) Hilbert space. In this sense the state space of quantum mechanics is explained and derived as a singular limit of a candidate superdeterministic successor theory. 

Indeed the same consideration applies to the classical limit at $\hbar=0$ in invariant set theory. At this limit, but only at this limit, the fractal helix collapses to a single state-space trajectory, as we would expect from any classical dynamical system based on deterministic differential equations or finite differences. By contrast the fractal gaps and discretised Hilbert space formalism applies for any finite value of $\hbar$ no matter how small. 

In conclusion, quantum mechanics is not `grudgingly accommodated' by the candidate superdeterministic theory. Rather, it is accommodated in exactly the same way that other great theories have been accommodated by their successor theories, from the time of Galileo and Newton to the present. Whether this explanation is `glorious', we leave the reader to decide. 

\section{Conclusions}

What is the difference between a deterministic theory and a superdeterministic theory? Both are deterministic in the sense that the future is determined by the past. However, in a superdeterminstic theory it is not \emph{necessarily} the case that, had the past been different in a certain way, the future would be different in a certain way. The latter invokes counterfactual reasoning - if I hadn't thrown the stone keeping everything else fixed, the window wouldn't have broken - and used to infer causality (if the counterfactual is true I must have caused the window to break). In a superdeterministic theory, quantum counterfactuals may or may not be inconsistent with the fundamental premises of the theory. In discussing Bell's theorem we have discussed counterfactual experiments that are consistent with the underlying rules of our superdeterministic model, and other counterfactual experiments that are inconsistent with our superdeterministic model. The latter helps explain why counterfactual causal reasoning can be so problematic in quantum physics. 

Let us return to the parody of superdeterminism as described in the Introduction. What's wrong with it? 

The first point that is wrong is the delineation into dynamical laws and initial conditions, as if these were independent of each other. Whilst this is the way we describe things in classical physics, it is not a property of the model described here. Instead we have described a model whose fundamental tenet is the idea that the universe is a deterministic system evolving on some fractal invariant set in cosmological state space. The geometry of the invariant set is fundamental. In this picture, one cannot vary initial conditions arbitrarily, keeping the dynamical laws fixed -- varying the initial conditions arbitrarily will take you off the invariant set to a point which is inconsistent with the dynamical laws. 

On top of this, properties of the fractal geometry of invariant sets are formally non-computable \cite{Dube:1993} \cite{Blum}. Hence it is not possible to have foreknowledge of whether a particular state lies on the invariant set or not. We cannot postulate a daemon who knows, at the time of the Big Bang, that Alice's choice of measurement setting depends on her grandmother's birthday. 

The second point that is wrong is the notion that there is only one initial state that could lead to the violation of Bell inequalities. The cardinality of the Cantor set is no less than the cardinality of the set of real numbers. There are infinitely many initial states that will lead to a violation of Bell inequalities \cite{HanceHossenfelder}. 

However, probably the most important misconception in the parody is in the notion of fine tuning and hence of conspiracy. It hardly needs saying but one cannot describe some theory as fine tuned without specifying a metric with respect to which the tuning is deemed fine. Ostrowsky's theorem \cite{Katok} tells us that, fundamentally, there are only two (norm-induced) classes of metric in mathematics: the Euclidean metric and the $p$-adic metric. Because of the close affinity of $p$-adic numbers to fractal geometry, the $p$-adic metric is a natural metric to use for a model based on fractal geometry in state space. As discussed, a point which does not lie on the invariant set (homeomorphic to the set $\mathbb Z_p$ of $p$-adic integers) is $p$-distant from a point on the invariant set, no matter how close these two points appear from a Euclidean perspective. Here we have assumed $p \gg 1$ in order that experimental results, consistent with quantum mechanics, are also consistent with our discretised Hilbert space. The idea that one can always do minimal `surgeries' on events in  space-time \cite{Pearl} in order to argue about counterfactual causality is false in invariant set theory for the simple reason that in the relevant $p$-adic norm these surgeries may not be minimal, they may be maximal.

Hence, if I want to argue the quantum equivalent of the argument that my throwing the stone caused the window to break, I should do so based on physics-based calculations in space-time (which has a causal structure), and not on potentially inconsistent counterfactual inferences based on surgeries in state space. As discussed, the analysis of Bell's theorem relative to our discretised Hilbert space makes no mention of causal structures in space-time. As such, it does not invoke or require any breakdown of the locally causal structure of relativity theory. Hence, to return to a key motivation for this work discussed in the Introduction, this superdeterministic model may better combine with general relativity theory than does quantum mechanics. 

It is worth concluding with a remark by Bell himself \cite{Bellb}:

\begin{quote}
Of course it might be that these reasonable ideas about physical randomisers are just wrong - for the purposes at hand. A theory may appear in which such [apparent] conspiracies inevitably occur, and these conspiracies may then seem more digestible than the non-localities of other theories. When that theory is announced I will not refuse to listen, either on methodological or other grounds. 
\end{quote}

It is a shame that modern-day commentators are not as open minded. 

\section*{Acknowledgements}
My thanks to Jonte Hance and Sabine Hossenfelder for many helpful comments on an early draft of this paper. This work has been supported by a Royal Society Research Professorship. 

\bibliography{mybibliography}
\end{document}